\documentclass[letterUSApaper,12pt,openright]{article}
\usepackage[T1]{fontenc}
\usepackage[applemac]{inputenc}
\usepackage{mathrsfs}
\usepackage{indentfirst}
\usepackage{amssymb,amsmath}
\usepackage{amsthm}
\usepackage{makeidx}
\usepackage{array}
\usepackage{graphicx}
\usepackage{latexsym}
\usepackage{caption}
\usepackage{fancyhdr}
\usepackage{bm}
\usepackage{appendix}
\title{Modeling and Pricing of Covariance and Correlation Swaps for Financial Markets with Semi-Markov Volatilities}
\author{
        Giovanni Salvi \\
                Department of Methods and Models for Economics, \\
                Territory and Finance MEMOTEF\\
		        `Sapienza' University of Rome, \\
		        Rome, Italy \\
            \and
        Anatoliy V. Swishchuk\footnote{Corresponding author} \\
        Department of Mathematics and Statistics, \\
        University of Calgary, \\
        Calgary, Alberta, Canada
}

\begin{document}
\setlength{\textheight}     {21.0  cm}
\setlength{\textwidth}      {16.0  cm}
\setlength{\baselineskip}   { 0.6  cm}
\setlength{\baselineskip}   { 23 pt}
\pagestyle{fancy}
\fancyhf{}
\fancyhead[LE,RO]{\bfseries\thepage}
\fancyhead[LO]{\bfseries\rightmark}
\fancyhead[RE]{\bfseries\leftmark}
\fancypagestyle{plain}{\fancyhead{} \renewcommand{\headrulewidth}{0pt}}

\newcommand{\beq}{\begin{equation}}
\newcommand{\eeq}{\end{equation}}
\newcommand{\bea}{\begin{eqnarray}}
\newcommand{\eea}{\end{eqnarray}}
\newcommand{\Prob}{\mathbb{P}}
\newcommand{\Exp}{\mathbb{E}}
\newcommand{\Alg}{\mathscr{F}}
\newcommand{\Real}{\mathbb{R}}
\newcommand{\Natural}{\mathbb{N}}
\newcommand{\Borel}{\mathscr{B}}
\newtheorem{teo}{Theorem}
\newtheorem{prop}{Proposition}
\newtheorem{rem}{Remark}
\newtheorem{lemma}{Lemma}
\newtheorem{cor}{Corollary}
\def\<{\langle}
\def\>{\rangle}
\thispagestyle{empty}

\maketitle

\begin{abstract}
In this paper, we model financial markets with semi-Markov volatilities and price covarinace and correlation swaps for this markets. Numerical evaluations of varinace, volatility, covarinace and correlations swaps with semi-Markov volatility are presented as well. The novelty of the paper lies in pricing of volatility swaps in closed form, and pricing of covariance and correlation swaps in a market with two risky assets.
\end{abstract}

\clearpage

\section*{Introduction}
One of the recent and new financial products are variance and voaltility swaps,
which are useful for volatility hedging and speculation. The market for variance
and volatility swaps has been growing, and many investment banks and other financial 
institutions are now actively quoting volatility swaps on various assets: stock
indexes, currencies, as well as commodities. A stock's volatility is the simpliest
measure of its riskiness or uncertainty. Formally, the volatility $\sigma_R$ is the annualized
standard deviation of the stock's returns during the period of interest, where the
subscript R denotes the observed or 'realized' volatility. Why trade volatility or
variance? As mentioned in [5], 'just as stock investors think they know something
about the direction of the stock market so we may think we have insight into the
level of future volatility. If we think current volatility is low, for the right price we
might want to take a position that profits if volatility increases'.

In this paper, we model financial markets with semi-Markov volatilities and price covarinace and correlation swaps for this markets. Numerical evaluations of varinace, volatility, covarinace and correlations swapswith semi-Markov volatility are presented as well. 

Volatility swaps are forward contracts on future realized stock volatility, variance swaps are similar contract on variance, the square of the future volatility, both these instruments provide an easy way for investors to gain exposure to the future level of volatility.
A stock's volatility is the simplest measure of its risk less or uncertainty. Formally, the volatility $\sigma_R$ is the annualized standard deviation of the stock's returns during the period of interest, where the subscript R denotes the observed or "realized" volatility.

The easy way to trade volatility is to use volatility swaps, sometimes called realized volatility forward contracts, because they provide pure exposure to volatility(and only to volatility).

A stock {\it volatility swap} is a forward contract on the annualized volatility. Its payoff at expiration is equal to 
$$
N(\sigma_R(S)-K_{vol}),
$$

where $\sigma_R(S)$ is the realized stock volatility (quoted in annual terms) over the life of contract,
 $$\sigma_R(S):=\sqrt{\frac{1}{T}\int_0^T\sigma^2_sds},$$
$\sigma_t$ is a stochastic stock volatility,
$K_{vol}$ is the annualized volatility delivery price, and N is the notional amount of the swap in dollar per annualized volatility point. The holder of a volatility swap at expiration receives N dollars for every point by which the stock's realized volatility $\sigma_R$ has exceeded the volatility delivery price $K_{vol}$. The holder is swapping a fixed volatility 
$K_{vol}$ for the actual (floating) future volatility $\sigma_R$. 
We note that usually $N=\alpha I,$
where $\alpha$ is a converting parameter such as 1 per volatility-square, 
and I is a long-short index (+1 for long and -1 for short).

Although options market participants talk of volatility, it is variance, or volatility squared, that has more fundamental significance (see Demeterfi {\it et al} \cite{DDKZ}).

{\it A variance swap} is a forward contract on annualized variance, the square of the realized volatility. Its payoff at expiration is equal to 
$$
N(\sigma_R^2(S)-K_{var}),
$$
where $\sigma_R^2(S)$ is the realized stock variance(quoted in annual terms) over the life of the contract, 
$$\sigma^2_R(S):=\frac{1}{T}\int_0^T\sigma^2_sds,$$
$K_{var}$ is the delivery price for variance, and N is the notional amount of the swap in dollars per annualized volatility point squared. The holder of variance swap at expiration receives N dollars for every point by which the stock's realized variance $\sigma_R^2(S)$ has exceeded the variance delivery price $K_{var}$. 

Therefore, pricing the variance swap reduces to calculating the realized volatility square. 

Valuing a variance forward contract or swap is no different from valuing any other derivative 
security. The value of a forward contract $P$ on future realized variance with strike price $K_{var}$ is the expected present value of the future payoff in the risk-neutral world:
$$P=E\{e^{-rT}(\sigma^2_R(S)-K_{var})\},$$
where $r$ is the risk-free discount rate corresponding to the expiration date $T,$ and $E$ denotes the expectation.

Thus, for calculating variance swaps we need to know only $E\{\sigma^2_R(S)\},$ namely, mean value of the underlying variance.

To calculate volatility swaps we need more. From Brockhaus and Long \cite{BrockLong} approximation (which is used the second order Taylor expansion for function $\sqrt{x}$) we have (see also Javaheri {\it et al} \cite{Java}):
$$E\{\sqrt{\sigma^2_R(S)}\}\approx\sqrt{E\{V\}}-\frac{Var\{V\}}{8E\{V\}^{3/2}},$$
where $V:=\sigma^2_R(S)$ and $\frac{Var\{V\}}{8E\{V\}^{3/2}}$ is the convexity adjustment.

Thus, to calculate volatility swaps we need both $E\{V\}$ and $Var\{V\}.$

The realised continuously sampled variance is defined in the following way:
$$
V:=Var(S):=\frac{1}{T}\int_0^T\sigma^2_tdt.
$$
Realised continuously sampled volatility is defined as follows:
$$
Vol(S):=\sqrt{Var(S)}=\sqrt{V}.
$$


Option dependent on exchange rate movements, such as those paying in a currency different from the underlying currency, have an 
exposure to movements of the correlation between the asset and the exchange rate, this risk may be eliminated by using covariance swap.
Variance and volatility swaps have been studied by Swishchuk \cite{Swish}. The novelty of this paper with respect to \cite{Swish} is that we calculate the volatility swap price explicitly, moreover we price covariance and correlation swap in a two risky assets market model.

A {\it covariance swap} is a covariance forward contact of the underlying rates $S^1$ and $S^2$ which payoff at expiration is equal to 
$$
N(Cov_R(S^1,S^2)-K_{cov}), 
$$
where $K_{cov}$ is a strike price, N is the notional amount, $Cov_R(S^1,S^2)$ is a covariance between two assets $S^1$ and $S^2$.

Logically, a {\it correlation swap} is a correlation forward contract of two underlying rates $S^1$ and $S^2$ which payoff at expiration is equal to:
$$N(Corr_R(S^1,S^2)-K_{corr}),$$
where $Corr(S^1,S^2)$ is a realized correlation of two underlying assets $S^1$ and $S^2 , K_{corr}$ is a strike price, N is the notional amount.

Pricing covariance swap, from a theoretical point of view, is similar to pricing variance swaps, since

$$
Cov_R(S^1,S^2)=1/4\{\sigma_R^2(S^1S^2)-\sigma_R^2(S^1/S^2)\} 
$$
where $S^1$ and $S^2$ are given two assets,
$\sigma_R^2(S)$ is a variance swap for underlying assets,
$Cov_R(S^1,S^2)$ is a realized covariance of the two underlying assets $S^1$ and $S^2$.

Thus, we need to know variances for $S^1S^2$ and for $S^1/S^2$ (see Section 4.2 for details).
Correlation $Corr_R(S^1,S^2)$ is defined as follows:
$$
Corr_R(S^1,S^2) = \frac {Cov_R(S^1,S^2)}
{\sqrt{\sigma_R^2(S^1)}\sqrt{\sigma_R^2(S^2)}},
$$
where $Cov_R(S^1,S^2)$ is defined above and $\sigma^2_R(S^1) $ in section 3.4. 

Given two assets $S^1_t$ and $S^2_t$ with $t\in [0,T],$ sampled on days $t_0=0<t_1<t_2<...<t_n=T$ between today and maturity $T,$ the log-return each asset is:
$R^j_i:=\log(\frac{S^j_{t_{i}}}{S^j_{t_{i-1}}}),  \quad i=1,2,...,n, \quad j=1,2.$ 

Covariance and correlation can be approximated by 
$$
Cov_n(S^1,S^2)=\frac{n}{(n-1)T}\sum_{i=1}^nR^1_iR^2_i
$$
and 
$$
Corr_n(S^1,S^2)=\frac{Cov_n(S^1,S^2)}{\sqrt{Var_n(S^1)}\sqrt{Var_n(S^2)}},
$$
respectively.

The literature devoted to the volatility derivatives is growing. We give here a short overview of the latest development in this area. 
The Non-Gaussian Ornstein-Uhlenbeck stochastic volatility model  was used by Benth et al. \cite{BGK} to study volatility and variance swaps. M. Broadie and A. Jain \cite{BJ} evaluated price and hedging strategy for volatility derivatives in the Heston square root stochastic volatility model and in \cite{BJ2} they compare result from various model in order to investigate the effect of jumps and discrete sampling on variance and volatility swaps. Pure jump process with independent increments return models were used by Carr et al \cite{CarrGMadanYor} to price derivatives written on realized variance, and subsequent development by Carr and Lee \cite{CarrLee1}. We also refer to Carr and Lee \cite{CarrLee2} for a good survey on volatility derivatives. Da Foneseca et al. \cite{DFIG} analyzed the influence of variance and covariance swap in a market by solving a portfolio optimization problem in a market with risky assets and volatility derivatives. Correlation swap price has been investigated by Bossu \cite{Bossu1} and \cite{Bossu2} for component of an equity index using statistical method. The paper \cite{DMV2} discusses the price of correlation risk for equity options. Pricing volatility swaps under Heston's model with regime-switching and pricing options under a generalized Markov-modulated jump-diffusion model are discussed in \cite{ElliottSC1} and \cite{ElliottSC2}, respectively.  The paper \cite{HRR} considers the pricing of a range of volatility derivatives, including volatility and variance swaps and swaptions. The pricing options on realized variance in the Heston model with jumps in returns and volatility is studied in \cite{Sepp}. An analytical closed-forms pricing of pseudo-variance, pseudo-volatility, pseudo-covariance and pseudo-correlation swaps is studied in \cite{Swish2}. The paper \cite{WFV} investigates the behaviour and hedging of discretely observed volatility derivatives.

The rest of the paper is organized as follows. Section 1 is devoted to the study of semi-Markov volatilites and their martingale characterization. Section 2 deals with pricing of variance and volatility swaps and numerical evaluation of them for markets with semi-Markov volatilities. Section 3 studies the pricing of covarinace and correlation swaps for a two risky assets in financial markets with semi-Markov volatilities. We also give here a numerical evaluation of these derivatives.  Appendix presents  a first order correction for a realized correlation.



\section{Martingale Representation of Semi-Markov Processes}

Let  $(\Omega, \Alg, (\Alg_{t})_{t \in \Real_{+}}, \Prob)$ be a filtered probability space, with a right-continuous filtration $(\Alg_{t})_{t \in \Real_{+}}$ and probability $\Prob$.

Let $(X,\mathscr{X})$ be a measurable space and
\bea
\label{Eq:SMkernel}
Q_{SM} (x,B,t) := P(x,B)G_{x}(t) \qquad \qquad \textrm{for} \ x \in X, B \in \mathscr{X}, t \in \Real_{+},
\eea
be a semi-Markov kernel. Let $(x_{n},\tau_{n}; n \in \Natural)$ be a $(X \times \Real_{+}, \mathscr{X} \otimes \Borel_{+})$-valued Markov renewal process with $Q_{SM}$ the associated kernel, that is
\bea
\Prob (x_{n+1} \in B , \tau_{n+1} - \tau_{n} \leq t \ | \ \Alg_{n}) = Q_{SM}(x_{n}, B, t).
\eea
Let define the process
\bea
\nu_{t} := sup \{ n \in \Natural : \tau_{n} \leq t \} 
\eea
that gives the number of jumps of the Markov renewal process in the time interval $(0,t]$ and
\bea
\theta_{n} := \tau_{n} - \tau_{n-1}
\eea
that gives the sojourn time of the Markov renewal process in the $n$-th visited state.
The semi-Markov process, associated with the Markov renewal process  $(x_{n},\tau_{n})_{n \in \Natural}$, is defined by
\bea
x_{t} := x_{\nu(t)} \qquad \qquad \textrm{for} \ t \in \Real_{+} .
\eea
Associated with the semi-Markov process it is possible to define some auxiliaries processes. We are interested in the backward recurrence time (or life-time) process defined by
\bea
\gamma(t) := t - \tau_{\nu(t)} \qquad \qquad \textrm{for} \  t \in \Real_{+} .
\eea
The next result characterizes backward recurrence time process (cf. Swishchuk \cite{Swish}).
\begin{prop}
The backward recurrence time $(\gamma(t))_{t}$ is a Markov process with generator
\bea
Q_{\gamma} f (t) = f'(t) + \lambda (t) [f(0) - f(t)] ,
\eea
where $\lambda (t) = -\frac{\overline{G_{x}}'(t)}{\overline{G_{x}}(t)}$, $\overline{G_{x}}(t) = 1 - G_{x}(t)$ and $Domain (Q_{\gamma}) = C^{1}(\Real_{+})$
\end{prop}
\begin{proof}
Let $t$ be the present time such that $\gamma(t) = t$, without loss of generality we can assume that $t < \tau_{1}$, then for $T>t$ we have
\bea
\Exp_{t} \{ f(\gamma(T)) \} = \Exp_{t} \{ f(\gamma(T)) \mathbb{I}_{\theta_{1}>T} \} + \Exp_{t} \{ f(\gamma(T)) \mathbb{I}_{\theta_{1}\leq T} \} .
\eea
Using the properties of conditional expectation we obtain
\bea
\begin{aligned}
\Exp_{t} \{ f(\gamma(T)) \} = & f(T) \frac{\overline{G}_{x}(T)}{\overline{G}_{x}(t)} + \frac{1}{\overline{G}_{x}(t)} \Exp \{ f(\gamma(T)) \mathbb{I}_{t < \theta_{1} \leq T} \} \\
= & f(T) \frac{\overline{G}_{x}(T)}{\overline{G}_{x}(t)} + \frac{1}{\overline{G}_{x}(t)} \int_{t}^{T} f(T-u) G_{x}'(u)du .
\end{aligned}
\eea
By adding and subtracting $f(t)$ in the integrand we get
\bea
\Exp_{t} \{ f(\gamma(T)) \} = f(T) \frac{\overline{G}_{x}(T)}{\overline{G}_{x}(t)} + \frac{1}{\overline{G}_{x}(t)} \int_{t}^{T} (f(T-u) - f(t)) G_{x}'(u)du + f(t)\frac{\overline{G}_{x}(t) - \overline{G}_{x}(T)}{\overline{G}_{x}(t)} , \nonumber
\eea
then
\bea
\Exp_{t} \{ f(\gamma(T)) \} - f(t) = (f(T)-f(t)) \frac{\overline{G}_{x}(T)}{\overline{G}_{x}(t)} + \frac{1}{\overline{G}_{x}(t)} \int_{t}^{T} (f(T-u) - f(t)) G_{x}'(u)du .
\eea
Now we recall the definition of the generator and using the above equation we have
\bea
Q_{\gamma} f(t) = \lim_{T \to t} \frac{\Exp_{t} \{ f(\gamma(T)) \} - f(t)}{T-t} =  f'(t) - \frac{\overline{G_{x}}'(t)}{\overline{G_{x}}(t)} [f(0) - f(t)] ,
\eea
this concludes the proof.
\end{proof}

\begin{rem}
As well known, semi-Markov process preserve the lost-memories property only at transition time, then $(x_{t})_{t}$ is not Markov. However, if we consider the joint process $(x_{t},\gamma(t))_{t \in \Real_{+}}$, we record at any instant the time already spent by the semi-Markov process in the present state, then it result that $(x_{t},\gamma(t))_{t \in \Real_{+}}$ is a Markov process.
\end{rem}

The following result allow us to associate a martingale to any Markov process we refer to Swishchuk \cite{Swish}, Elliott and Swishchuk \cite{ElliottSwish} for details.

\begin{prop} 
\label{Prop:AssMg}
Let $(x_{t})_{t \in \Real_{+}}$ be a Markov process with generator $Q$ and $f \in$ Domain $(Q)$, then
\bea
m^{f}_{t} := f(x_{t}) - f(x_{0}) - \int_{0}^{t} Q f(x_{s}) ds
\eea
is a zero-mean martingale with respect to $\Alg_{t} := \sigma \{ y(s);0 \leq s \leq t \}$.
\end{prop}

Let us evaluate the quadratic variation of this martingale (see Swishchuk \cite{Swish}, Elliott and Swishchuk \cite{ElliottSwish}).

\begin{prop}
\label{Prop:QuadVar}
Let $(x_{t})_{t \in \Real_{+}}$ be a Markov process with generator $Q$, $f \in$ Domain $(Q)$ and $(m^{f}_{t})_{t \in \Real_{+}}$ its associated martingale, then 
\bea
\langle m^{f} \rangle_{t} := \int_{0}^{t} [Q f^{2}(x_{s}) - 2 f(x_{s}) Q f(x_{s}) ] ds
\eea
is the quadratic variation of $m^{f}$.
\end{prop}

The next result concern the evaluation of the quadratic covariation, for details and proof we refer to Salvi and Swishchuk \cite{SalviSwish}.

\begin{prop}
\label{Prop:QuadCovar}
Let $(x_{t})_{t \in \Real_{+}}$ be a Markov process with generator $Q$,  $f,g \in$ Domain $(Q)$ such that $fg \in$ Domain $(Q)$.  Denote by $(m^{f}_{t})_{t \in \Real_{+}}, (m^{g}_{t})_{t \in \Real_{+}}$ their associated martingale. Then
\bea
\langle f(x_{\cdot}), g(x_{\cdot}) \rangle_{t} := \int_{0}^{t} \{ Q (f(x_{s})g(x_{s})) - [ g(x_{s}) Q f(x_{s}) + f(x_{s}) Q g(x_{s}) ] \} ds
\eea
is the quadratic covariation of $f$ and $g$.
\end{prop}

We would like to study the martingale associated to the Markov process $(x_{t},\gamma(t))_{t \in \Real_{+}}$ and its generator the following statement concern this task, it is a direct application of proposition \ref{Prop:AssMg}.

\begin{lemma}
(Swishchuk \cite{Swish}) Let $(x_{t})_{t \in \Real_{+}}$ be a semi-Markov process with kernel $Q_{SM}$ defined in (\ref{Eq:SMkernel}). Then, the process
\bea
m^{f}_{t} := f(x_{t},\gamma(t)) - \int_{0}^{t} Q f(x_{s},\gamma(s)) ds
\eea
is a martingale with respect to the filtration $\Alg_{t} := \sigma \{ x_{s}, \nu_{s}; 0 \leq s \leq t \}$, where $Q$ is the generator of the Markov process $(x_{t},\gamma(t))_{t \in \Real_{+}}$ given by
\bea
Q f(x,t) = \frac{d f}{dt}(x,t) + \frac{g_{x}(t)}{\overline{G}_{x}(t)} \int_{X}P(x,dy) [f(y,0) - f(x,t)] ,
\eea
here $g_{x}(t) = \frac{dG_{x}(t)}{dt}$.
\end{lemma}

The following statement follows directly from proposition \ref{Prop:QuadVar} and it allows us to evaluate the quadratic variation of the martingale associated with $(x_{t},\gamma(t))_{t \in \Real_{+}}$.

\begin{lemma}
\label{Lem:GenQ}
Let $(x_{t})_{t \in \Real_{+}}$ be a semi-Markov process with kernel $Q_{SM}$, $(x_{t},\gamma(t))_{t \in \Real_{+}}$ is a Markov process with generator Q, $f \in$ Domain $(Q)$ and $(m^{f}_{t})_{t \in \Real_{+}}$ its associated martingale, then 
\bea
\langle m^{f} \rangle_{t} := \int_{0}^{t} [Q f^{2}(x_{s},\gamma(s)) - 2 f(x_{s},\gamma(s)) Q f(x_{s},\gamma(s)) ] ds
\eea
is the quadratic variation of $m^{f}$.
\end{lemma}

\section{Variance and Volatility Swaps for Financial Markets with Semi-Markov Stochastic Volatilities}

Let consider a Market model with only two securities, the risk free bond and the stock. Let suppose that the stock price $(S_{t})_{t \in \Real_{+}}$ satisfies the following stochastic differential equation
\bea
dS_{t} = S_{t} (r dt + \sigma(x_{t},\gamma(t)) dw_{t})
\eea
where $w$ is a standard Wiener process independent of $(x,\gamma)$. 
We are interested in studying the property of the volatility $\sigma(x,\gamma)$.
Salvi and Swishchuk \cite{SalviSwish} have studied properties of volatility modulated by a Markov process here we would like to generalize their work to the semi-Markov case. First of all we study the second moment of the volatility, see Swishchuk \cite{Swish} for details and proof.

\begin{prop}
\label{Prop:Expsigma2}
Suppose that $\sigma \in$ Domain $(Q)$. Then 
\bea
\Exp \{ \sigma^{2}(x_{t},\gamma(t)) |  \Alg_{u} \} =  \sigma^{2}(x_{u},\gamma(u)) + \int_{u}^{t} Q \Exp \{ \sigma^{2}(x_{s},\gamma(s)) | \Alg_{u} \}  ds
\eea
for any $0 \leq u \leq t$.
\end{prop}

\begin{rem}
\label{Rem:Expsigma2}
From proposition \ref{Prop:Expsigma2}, we can directly solve the equation for $\Exp \{ \sigma^{2}(x_{t},\gamma(t)) |  \Alg_{u} \}$ and we obtain
\bea
\Exp \{ \sigma^{2}(x_{t},\gamma(t)) |  \Alg_{u} \} = e^{(t-u)Q} \sigma^{2}(x_{u},\gamma(u))
\eea
for any $0 \leq u \leq t$. \\
\end{rem}

In this semi-Markov modulated model we are able to evaluate high order moment of volatility, e.g next result concern the fourth moment.

\begin{prop}
\label{Prop:Expsigma4}
Suppose that $\sigma^{2} \in$ Domain $(Q)$. 
\bea
\Exp \{ \sigma^{4}(x_{t},\gamma(t)) \} = e^{tQ} \sigma^{4}(x,0) \qquad \qquad \textrm{for} \ t \in \Real_{+} .
\eea
\end{prop}
\begin{proof} If  $\sigma^{2} \in$ Domain $(Q)$, then from Proposition \ref{Prop:AssMg} we have that
\bea
\label{Eq:sigma2Mg}
m^{\sigma^{2}}_{t} := \sigma^{2}(x_{t},\gamma(t)) - \sigma^{2}(x,0) - \int_{0}^{t} Q\sigma^{2}(x_{s},\gamma(s)) ds
\eea
is a zero-mean martingale with respect to $\Alg_{t} := \sigma \{ x_{s},\gamma(s);0 \leq s \leq t \}$, where $Q$ is the infinitesimal generator of the Markov process $(x_{t},\gamma(t))_{t \in [0,T]}$  (cf. lemma \ref{Lem:GenQ}).  Then $\sigma^{2}$ satisfies the following stochastic differential equation
\bea
d\sigma^{2}(x_{t},\gamma(t)) = Q\sigma^{2}(x_{t},\gamma(t)) dt + d m^{\sigma^{2}}_{t} .
\eea
By applying the Ito's Lemma we obtain
\bea
d\sigma^{4}(x_{t},\gamma(t)) = 2 \sigma^{2}(x_{t},\gamma(t)) d\sigma^{2}(x_{t},\gamma(t)) + d \langle m^{\sigma^{2}}_{\cdot} \rangle_{t} .
\eea
We note that, the quadrating variation of $m^{\sigma^{2}}$ is (see Propositon \ref{Prop:QuadVar}) given by
\bea
 \langle m^{\sigma^{2}} \rangle_{t} := \int_{0}^{t} [Q \sigma^{4}(x_{s},\gamma(s)) - 2 \sigma^{2}(x_{s},\gamma(s)) Q \sigma^{2}(x_{s},\gamma(s)) ] ds .
\eea
Substituting the expression for the quadratic variation $\langle m^{\sigma^{2}} \rangle$ and for $\sigma^{2}$, we obtain that $\sigma^{4}$ satisfies the following stochastic differential equation
\bea
d\sigma^{4}(x_{t},\gamma(t)) =  Q \sigma^{4}(x_{t},\gamma(t)) dt + 2 \sigma^{2}(x_{t},\gamma(t))d m^{\sigma^{2}}_{t} .
\eea
By tacking the expectation of both side of these equation we obtain
\bea
\Exp \{ \sigma^{4}(x_{t},\gamma(t)) \} =  \sigma^{4}(x,0) + \int_{0}^{t} Q \Exp \{ \sigma^{4}(x_{s},\gamma(s)) \} ds .
\eea
Solving this differential equation we finally get
\bea
\Exp \{ \sigma^{4}(x_{t},\gamma(t)) \} = e^{tQ} \sigma^{4}(x,0) .
\eea
\end{proof}

It is known that the market model with semi-Markov stochastic volatility is incomplete, see Swishchuk \cite{Swish}.  In order to price the future contracts we will use the minimal martingale measure, we refer to Swishchuk \cite{Swish} for the details.
Let us now focus on the evaluation of the price of variance and volatility Swaps.

\subsection{Pricing of Variance Swaps}

Let's start from the more straightforward variance swap. Variance swaps are forward contract on future realized level of variance. The payoff of a variance swap with expiration date $T$ is given by
\bea
N(\sigma^{2}_{R} (x) - K_{var})
\eea
where $\sigma^{2}_{R} (x)$ is the realized stock variance over the life of the contract
\bea
\sigma^{2}_{R} (x) := \frac{1}{T} \int_{0}^{T}  \sigma^{2}(x_{s},\gamma(s)) ds ,
\eea
$K_{var}$ is the strike price for variance and $N$ is the notional amounts of dollars per annualized variance point, we will assume that $N = 1$ just for sake of simplicity. The price of the variance swap is the expected present value of the payoff in the risk-neutral world 
\bea
P_{var} (x) = \Exp \{ e^{-rT} (\sigma^{2}_{R} (x) - K_{var}) \} .
\eea
The following result concerns the evaluation of a variance swap in this semi-Markov volatility model. We refer to Swishchuk \cite{Swish} for details and proof.
\begin{teo}
\label{teo:VarSwap}
(Swishchuk \cite{Swish}) The present value of a variance swap for semi-Markov stochastic volatility is
\bea
P_{var} (x) = e^{-rT} \left\{ \frac{1}{T} \int_{0}^{T} (e^{tQ} \sigma^{2}(x,0) - K_{var}) dt \right\}
\eea
where $Q$ is the generator of $(x_{t},\gamma(t))_{t}$, that is
\bea
Q f(x,t) = \frac{d f}{dt}(x,t) + \frac{g_{x}(t)}{\overline{G}_{x}(t)} \int_{X}P(x,dy) [f(y,0) - f(x,t)] .
\eea
\end{teo}

\subsection{Pricing of Volatility Swaps}

 Volatility swaps are forward contract on future realized level of volatility. The payoff of a volatility swap with maturity $T$ is given by
\bea
N(\sigma_{R} (x) - K_{vol})
\eea
where $\sigma_{R} (x)$ is the realized stock volatility over the life of the contract
\bea
\sigma_{R} (x) := \sqrt{\frac{1}{T} \int_{0}^{T}  \sigma^{2}(x_{s},\gamma(s)) ds} ,
\eea
$K_{vol}$ is the strike price for volatility and $N$ is the notional amounts of dollars per annualized volatility point, as before we will assume that $N=1$. 
The price of the volatility swap is the expected present value of the payoff in the risk-neutral world 
\bea
P_{vol} (x) = \Exp \{ e^{-rT} (\sigma_{R} (x) - K_{vol}) \} .
\eea
In order to evaluate the volatility swaps we need to know the expected value of the square root of the variance, but unfortunately, in general we are not able to evaluate analytically this expected value.
Then in order to obtain a close formula for the price of volatility swaps we have to make an approximation. Using the same approach of the Markov case (see also Brockhaus and Long \cite{BrockLong} and Javaheri at al. \cite{Java}), from the second order Taylor expansion we have
\bea
\Exp \{ \sqrt{\sigma^{2}_{R} (x)} \} \approx \sqrt{\Exp \{ \sigma^{2}_{R} (x) \} } - \frac{Var \{ \sigma^{2}_{R} (x) \}}{8 \Exp \{ \sigma^{2}_{R} (x) \}^{3/2}} .
\eea
Then, to evaluate the volatility swap price we have to know both expectation and variance of $\sigma^{2}_{R} (x)$. The next result gives an explicit representation of the price of a volatility swap approximated to the second order for this semi-Markov volatility model.

\begin{teo}
\label{Th:VolSwap}
The value of a volatility swap for semi-Markov stochastic volatility is
\bea
 P_{vol} (x) \approx & e^{-rT} \left\{ \sqrt{\frac{1}{T} \int_{0}^{T}  e^{tQ} \sigma^{2}(x,0) dt} - \frac{Var \{ \sigma^{2}_{R} (x) \}}{8 \left( \frac{1}{T} \int_{0}^{T} e^{tQ} \sigma^{2}(x,0) dt \right)^{3/2}} - K_{vol} \right\} \nonumber
\eea
where the variance is given by
\bea
Var \{ \sigma^{2}_{R} (x) \} = \frac{2}{T^{2}} \int^{T}_{0} \int^{t}_{0} \left\{ e^{sQ} \left[ \sigma^{2}(x,0) e^{(t-s)Q} \sigma^{2}(x,0) \right] - \left[ e^{tQ} \sigma^{2}(x,0) \right] \left[ e^{sQ} \sigma^{2}(x,0) \right] \right\} dsdt , \nonumber
\eea
and $Q$ is the generator of $(x_{t},\gamma(t))_{t}$, that is
\bea
Q f(x,t) = \frac{d f}{dt}(x,t) + \frac{g_{x}(t)}{\overline{G}_{x}(t)} \int_{X}P(x,dy) [f(y,0) - f(x,t)] .
\eea
\end{teo}
\begin{proof}
We have already obtained the expectation of the realized variance,
\bea
\Exp \{ \sigma^{2}_{R} (x) \} = \frac{1}{T} \int_{0}^{T}  e^{tQ} \sigma^{2}(x,0) dt
\eea
then it remains to prove that
\bea
Var \{ \sigma^{2}_{R} (x) \} = \frac{2}{T^{2}} \int^{T}_{0} \int^{t}_{0} \left\{ e^{sQ} \left[ \sigma^{2}(x,0) e^{tQ} \sigma^{2}(x,0) \right] - \left[ e^{tQ} \sigma^{2}(x,0) \right] \left[ e^{sQ} \sigma^{2}(x,0) \right] \right\} dsdt . \nonumber
\eea
The variance is, from the definition,  given by
\bea
Var \{ \sigma^{2}_{R} (x) \} = \Exp \{ [\sigma^{2}_{R} (x) -  \Exp \{ \sigma^{2}_{R} (x) \}]^{2}\} ,
\eea
using the definition of realized variance, and Fubini theorem, we have
\bea
& & Var \{ \sigma^{2}_{R} (x) \}  =  \Exp \left\{ \left[ \frac{1}{T} \int^{T}_{0} \sigma^{2}(x_{t},\gamma(t)) dt - \frac{1}{T} \int^{T}_{0} \Exp \{ \sigma^{2}(x_{t},\gamma(t)) \} dt \right]^{2} \right\} \nonumber \\
& = & \Exp \left\{ \left[ \frac{1}{T} \int^{T}_{0} \left( \sigma^{2}(x_{t},\gamma(t)) - \Exp \{ \sigma^{2}(x_{t},\gamma(t)) \} \right) dt \right]^{2} \right\}  \\
& = & \frac{1}{T^{2}} \int^{T}_{0} \int^{T}_{0} \Exp \left\{ [ \sigma^{2}(x_{t},\gamma(t)) - \Exp \{ \sigma^{2}(x_{t},\gamma(t)) \} ] [ \sigma^{2}(x_{s},\gamma(s)) - \Exp \{ \sigma^{2}(x_{s},\gamma(s)) \} ]  \right\} ds dt\nonumber
\eea
We note that the integrand is symmetric in the exchange of s and t. 
We can divide the integration on the plan in two areas above and below the graph of t=s, thanks to the symmetry the contribution on the two parts is the same. Then we obtain
\bea
\begin{aligned}
\hspace{-20pt} Var \{ \sigma^{2}_{R} (x) \} = & \frac{2}{T^{2}} \int^{T}_{0} \int^{t}_{0} \Exp \{ [ \sigma^{2}(x_{t},\gamma(t)) - \Exp \{ \sigma^{2}(x_{t},\gamma(t)) \} ] [ \sigma^{2}(x_{s},\gamma(s)) - \Exp \{ \sigma^{2}(x_{s},\gamma(s)) \} ] \} ds dt \\
= & \frac{2}{T^{2}} \int^{T}_{0} \int^{t}_{0} \left[ \Exp \{ \sigma^{2}(x_{t},\gamma(t)) \sigma^{2}(x_{s},\gamma(s)) \} - \Exp \{ \sigma^{2}(x_{t},\gamma(t)) \}  \Exp \{ \sigma^{2}(x_{s},\gamma(s)) \} \right] dsdt .  \nonumber
\end{aligned}
\eea
We would like to stress that, in this representation, the integration set is such that the inequality $s \leq t$ holds true.  Using the property of conditional expectation and proposition \ref{Prop:Expsigma2}, we have
\bea
\hspace{-20pt}Var \{ \sigma^{2}_{R} (x) \} = \frac{2}{T^{2}} \int^{T}_{0} \int^{t}_{0} \left[ \Exp \{ \sigma^{2}(x_{s},\gamma(s)) \Exp \{ \sigma^{2}(x_{t},\gamma(t)) | \Alg_{s} \} \} - \left( e^{tQ} \sigma^{2}(x,0) \right) \left( e^{sQ} \sigma^{2}(x,0) \right) \right] dsdt . \nonumber
\eea
The process $(x_{t},\gamma(t))_{t}$ is Markov, then using remark \ref{Rem:Expsigma2} the conditional expected value in the integrand, can be expressed as
\bea
\Exp \{ \sigma^{2}(x_{t},\gamma(t)) | \Alg_{s} \} = e^{(t-s)Q} \sigma^{2}(x_{s},\gamma(s)) =: g(x_{s},\gamma(s)) .
\eea
Thus, the variance becomes
\bea
Var \{ \sigma^{2}_{R} (x) \} = \frac{2}{T^{2}} \int^{T}_{0} \int^{t}_{0} \left[ \Exp \{ \sigma^{2}(x_{s},\gamma(s)) g(x_{s},\gamma(s))\} - \left( e^{tQ} \sigma^{2}(x,0) \right) \left( e^{sQ} \sigma^{2}(x,0) \right) \right] dsdt . \nonumber
\eea
Solving the expectation on the right hand side we obtain
\bea
Var \{ \sigma^{2}_{R} (x) \} = \frac{2}{T^{2}} \int^{T}_{0} \int^{t}_{0} \left[ e^{sQ} \left( \sigma^{2}(x,0) g(x,0) \right) - \left( e^{tQ} \sigma^{2}(x,0) \right) \left( e^{sQ} \sigma^{2}(x,0) \right) \right] dsdt . \nonumber
\eea
We notice that function g evaluated in $(x,0)$ is simply given by
\bea
g(x,0) = e^{(t-s)Q} \sigma^{2}(x,0) ,
\eea
then substituting in the previous formula the variance finally becomes
\bea
Var \{ \sigma^{2}_{R} (x) \} = \frac{2}{T^{2}} \int^{T}_{0} \int^{t}_{0} \left\{ e^{sQ} \left[ \sigma^{2}(x,0) e^{(t-s)Q} \sigma^{2}(x,0) \right] - \left[ e^{tQ} \sigma^{2}(x,0) \right] \left[ e^{sQ} \sigma^{2}(x,0) \right] \right\} dsdt . \nonumber
\eea
\end{proof}

\subsection{Numerical Evaluation of Variance and Volatility Swaps with Semi-Markov Volatility}

In the application when we attempt to evaluate the price of a variance or a volatility swaps we have to deal with numerical problems. The family of exponential operators $(e^{tQ})_{t}$ involved in theorems \ref{teo:VarSwap} and \ref{Th:VolSwap} for the semi-Markov stochastic volatility model is usual difficult to evaluate from the numerical point of view. To solve this problem, we first look to the following identity
\bea
e^{tQ} f(\cdot) = \sum_{n=0}^{\infty} \frac{(tQ)^{n}}{n!} f(\cdot) ,
\eea
for any function $f \in Domain(Q)$. This identity allow us to obtain the operator $(e^{tQ})_{t}$ at any order of approximation. For example, for $n=1$, we obtain
\bea
e^{tQ}  f(\cdot) \approx (I + tQ) f(\cdot) 
\eea
where $I$ is an identity operator. At this order of approximation we allow semi-Markov process to make at most one transition during the life time of contract. If we think to semi-Markov process as a macroeconomic factor this can be plausible. However we can always evaluate the error in this approximation using the subsequent orders.
Using the first order approximation the variance swap price becomes
\bea
\begin{aligned}
P_{var} (x) \approx & e^{-rT} \left\{ \frac{1}{T} \int_{0}^{T}  (I + tQ) \sigma^{2}(x,0) dt - K_{var} \right\} \\
= & e^{-rT} \left\{ \sigma^{2}(x,0) + \frac{T}{2}Q \sigma^{2}(x,0) - K_{var} \right\} .
\end{aligned}
\eea
Using the same approximation the volatility swap price can be expressed as
\begin{eqnarray*}
\begin{aligned}
P_{vol} (x) \approx & e^{-rT} \left\{ \sqrt{\frac{1}{T} \int_{0}^{T}  (I + tQ) \sigma^{2}(x,0) dt} - \frac{Var \{ \sigma^{2}_{R} (x) \}}{8 \left( \frac{1}{T} \int_{0}^{T} (I + tQ) \sigma^{2}(x,0) dt \right)^{3/2}} - K_{vol} \right\} \\
= & e^{-rT} \left\{ \sqrt{ \sigma^{2}(x,0) + \frac{T}{2}Q \sigma^{2}(x,0)} - \frac{Var \{ \sigma^{2}_{R} (x) \}}{8\left( \sigma^{2}(x,0) + \frac{T}{2}Q \sigma^{2}(x,0) \right)^{3/2}} -  K_{vol} \right\} .
\end{aligned}
\end{eqnarray*}
Here, the variance of realized volatility is given by
\bea
\begin{aligned}
\hspace{-20pt}Var \{ \sigma^{2}_{R} (x) \} \approx & \frac{2}{T^{2}} \int^{T}_{0} \int^{t}_{0} \left\{ (I + sQ) [ \sigma^{2}(x,0) (I + (t-s)Q) \sigma^{2}(x,0) ] \right. \\
- & \left. [(I + tQ) \sigma^{2}(x,0)] [(I + sQ) \sigma^{2}(x,0)] \right\} ds dt .
\end{aligned}
\eea
Solving the product and keeping only the terms up to the first order in $Q$, we obtain
\bea
\begin{aligned}
Var \{ \sigma^{2}_{R} (x) \} \approx & \frac{2}{T^{2}} \int^{T}_{0} \int^{t}_{0} \{ sQ \sigma^{4}(x,0) - 2\sigma^{2}(x,0) sQ \sigma^{2}(x,0) \} ds dt  \\
= & \frac{T}{3} \left\{ Q \sigma^{4}(x,0) - 2\sigma^{2}(x,0) Q \sigma^{2}(x,0) \right\} .
\end{aligned}
\eea
Finally, at the first order of approximation in $Q$ the volatility swap price becomes
\begin{eqnarray*}
\begin{aligned}
P_{vol} (x) \approx e^{-rT} \left\{ \sqrt{ \sigma^{2}(x,0) + \frac{T}{2}Q \sigma^{2}(x,0)} - \frac{T [Q \sigma^{4}(x,0) - 2\sigma^{2}(x,0) Q \sigma^{2}(x,0)]}{24\left( \sigma^{2}(x,0) + \frac{T}{2}Q \sigma^{2}(x,0) \right)^{3/2}} -  K_{vol} \right\} .
\end{aligned}
\end{eqnarray*}

\section{Covariance and Correlation Swaps for a Two Risky Assets in Financial markets with Semi-Markov Stochastic Volatilities}

Let's consider now a market model with two risky assets and one risk free bond. Let's assume that the risky assets are satisfying the following stochastic differential equations
\bea
\left\{ \begin{matrix}
dS^{(1)}_{t} = S^{(1)}_{t} (\mu_{t}^{(1)}dt + \sigma^{(1)}(x_{t},\gamma(t))dw^{(1)}_{t}) \\
 \\
dS^{(2)}_{t} = S^{(2)}_{t} (\mu_{t}^{(2)}dt + \sigma^{(2)}(x_{t},\gamma(t))dw^{(2)}_{t})
\end{matrix}
\right.
\eea
where $\mu^{(1)}, \mu^{(2)}$ are deterministic functions of time, $(w^{(1)}_{t})_{t}$ and $(w^{(2)}_{t})_{t}$ are standard Wiener processes with quadratic covariance given by
\bea
d[w^{(1)}_{t},w^{(2)}_{t}] = \rho_{t} dt ,
\eea 
here $\rho_{t}$ is a deterministic function and $(w^{(1)}_{t})_{t}, (w^{(2)}_{t})_{t}$ are independent of $(x,\gamma)$.  \\
\noindent
In this model it is worth to study the covariance and the correlation swaps between the two risky assets.

\subsection{Pricing of Covariance Swaps}

A covariance swap is a covariance forward contract on the underlying assets $S^{(1)}$ and $S^{(2)}$ which payoff at maturity is equal to
\bea
N(Cov_{R}(S^{(1)},S^{(2)}) - K_{cov})
\eea
 where $K_{cov}$ is a strike reference value, $N$ is the notional amount and $Cov_{R}(S^{(1)},S^{(2)})$ is the realized covariance of the two assets $S^{(1)}$ and $S^{(2)}$ given by
 \bea
 Cov_{R}(S^{(1)},S^{(2)}) = \frac{1}{T} [\ln S^{(1)}_{T}, \ln S^{(2)}_{T}] = \frac{1}{T} \int_{0}^{T} \rho_{t} \sigma^{(1)}(x_{t},\gamma(t)) \sigma^{(2)}(x_{t},\gamma(t)) dt .
 \eea
 The price of the covariance swap is the expected present value of the payoff in the risk neutral world
\bea
P_{cov}(x) = \Exp \{ e^{-rT} (Cov_{R}(S^{(1)},S^{(2)}) - K_{cov}) \} ,
\eea
here we set $N=1$. The next result provides us an explicit representation of the covariance swap price.
\begin{teo}
\label{teo:CovSwap}
The value of a covariance swap for semi-Markov stochastic volatility is
\bea
P_{cov}(x) = e^{-rT} \left\{ \frac{1}{T} \int_{0}^{T} \rho_{t} e^{tQ} [\sigma^{(1)}(x,0) \sigma^{(2)}(x,0)] dt - K_{cov} \right\} ,
\eea
where $Q$ is the generator of $(x_{t},\gamma(t))_{t}$, that is
\bea
Q f(x,t) = \frac{d f}{dt}(x,t) + \frac{g_{x}(t)}{\overline{G}_{x}(t)} \int_{X}P(x,dy) [f(y,0) - f(x,t)] .
\eea
\end{teo}
\begin{proof}
To evaluate the price of covariance swap we need to know
\bea
\Exp \{ Cov_{R}(S^{(1)},S^{(2)}) \} = \frac{1}{T} \int_{0}^{T} \rho_{t} \Exp \{ \sigma^{(1)}(x_{t},\gamma(t)) \sigma^{(2)}(x_{t},\gamma(t)) \} dt .
\eea
It remains to prove that
\bea
\Exp \{ \sigma^{(1)}(x_{t},\gamma(t)) \sigma^{(2)}(x_{t},\gamma(t)) \} = e^{tQ} [\sigma^{(1)}(x,0) \sigma^{(2)}(x,0)] .
\eea
By applying the Ito's lemma we have
\bea
\label{eq:teo3.1}
\begin{aligned}
d(\sigma^{(1)}(x_{t},\gamma(t)) \sigma^{(2)}(x_{t},\gamma(t))) = & \sigma^{(1)}(x_{t},\gamma(t)) d\sigma^{(2)}(x_{t},\gamma(t)) + \sigma^{(2)}(x_{t},\gamma(t)) d\sigma^{(1)}(x_{t},\gamma(t)) \\
+& d \langle \sigma^{(1)}(x_{\cdot},\gamma(\cdot)), \sigma^{(2)}(x_{\cdot},\gamma(\cdot)) \rangle_{t} .
\end{aligned}
\eea
Using proposition \ref{Prop:QuadCovar} we obtain
\bea
\label{eq:teo3.2}
\begin{aligned}
 d \langle \sigma^{(1)}(x_{\cdot},\gamma(\cdot)), \sigma^{(2)}(x_{\cdot},\gamma(\cdot)) \rangle_{t} =  &
 Q (\sigma^{(1)}(x_{t},\gamma(t)) \sigma^{(2)}(x_{t},\gamma(t))) dt  \\
 - & [ \sigma^{(1)}(x_{t},\gamma(t)) Q \sigma^{(2)}(x_{t},\gamma(t)) + \sigma^{(2)}(x_{t},\gamma(t)) Q \sigma^{(1)}(x_{t},\gamma(t))] dt .
 \end{aligned}
\eea
Furthermore we have
\bea
\label{eq:teo3.3}
d  \sigma^{(i)}(x_{t},\gamma(t)) = Q  \sigma^{(i)}(x_{t},\gamma(t)) dt + dm^{ \sigma^{(i)}} \qquad \qquad i = 1,2 .
\eea
Substituting (\ref{eq:teo3.2}) and (\ref{eq:teo3.3}) in equation (\ref{eq:teo3.1}) we get
\bea
 \begin{aligned}
d(\sigma^{(1)}(x_{t},\gamma(t)) \sigma^{(2)}(x_{t},\gamma(t))) = & Q (\sigma^{(1)}(x_{t},\gamma(t)) \sigma^{(2)}(x_{t},\gamma(t))) dt  \\
+ & \sigma^{(1)}(x_{t},\gamma(t)) dm^{\sigma^{(2)}} + \sigma^{(2)}(x_{t},\gamma(t)) dm^{\sigma^{(1)}}.
\end{aligned}
\eea
Taking the expectation on both side we can rewrite the above equation as
\bea
\begin{aligned}
\Exp \{ \sigma^{(1)}(x_{t},\gamma(t)) \sigma^{(2)}(x_{t},\gamma(t)) \} =  &\sigma^{(1)}(x,0) \sigma^{(2)}(x,0) \\
+ & \int_{0}^{t} Q \Exp \{ \sigma^{(1)}(x_{s},\gamma(s)) \sigma^{(2)}(x_{s},\gamma(s)) \} dt .
\end{aligned}
\eea
Solving this differential equation we obtain
\bea
\Exp \{ \sigma^{(1)}(x_{t},\gamma(t)) \sigma^{(2)}(x_{t},\gamma(t)) \} = e^{tQ} [\sigma^{(1)}(x,0) \sigma^{(2)}(x,0)] ,
\eea
this conclude the proof.
\end{proof}

\subsection{Pricing of Correlation Swaps}

A correlation swap is a forward contract on the correlation between the underlying assets $S^{1}$ and $S^{2}$ which payoff at maturity is equal to
\bea
N(Corr_{R}(S^{1},S^{2}) - K_{corr})
\eea
where $K_{corr}$ is a strike reference level, $N$ is the notional amount and $Corr_{R}(S^{1},S^{2})$ is the realized correlation defined by
\bea
Corr_{R}(S^{1},S^{2}) = \frac{Cov_{R}(S^{1},S^{2})}{\sqrt{\sigma^{(1)^{2}}_{R}(x)}\sqrt{\sigma^{(2)^{2}}_{R} (x)}} ,
\eea
here the realized variance is given by
\bea
\sigma^{(i)^{2}}_{R}(x) = \frac{1}{T} \int_{0}^{T} (\sigma^{(i)}(x_{t},\gamma(t)))^{2}dt \qquad \qquad i = 1,2 .
\eea
The price of the correlation swap is the expected present value of the payoff in the risk neutral world, that is
\bea
P_{corr}(x) = \Exp \{ e^{-rT} (Corr_{R}(S^{1},S^{2}) - K_{corr}) \}
\eea
where we set $N=1$ for simplicity. Unfortunately the expected value of $Corr_{R}(S^{1},S^{2})$ is not known analytically. Thus, in order to obtain an explicit formula for the correlation swap price, we have to make some approximation.
\subsubsection{Correlation Swap made simple}
\label{Sec:CorrSimple}
First of all, let introduce the following notations
\bea
\begin{aligned}
X = & Cov_{R}(S^{1},S^{2}) \\
Y = & \sigma^{(1)^{2}}_{R}(x) \\
Z = & \sigma^{(2)^{2}}_{R}(x) ,
\end{aligned}
\eea
and with the subscript 0 we will denote the expected value of the above random variables. 
Following the approach frequently used for the volatility swap, we would like to approximate the square root of $Y$ and $Z$ at the first order as follows
\bea
\begin{aligned}
\sqrt{Y} \approx & \sqrt{Y_{0}} + \frac{Y - Y_{0}}{2\sqrt{Y_{0}}} \\
\sqrt{Z} \approx & \sqrt{Z_{0}} + \frac{Z - Z_{0}}{2\sqrt{Z_{0}}} .
\end{aligned}
\eea
The realized correlation can now be approximated by
\bea
Corr_{R}(S^{1},S^{2}) \approx \frac{X}{\left( \sqrt{Y_{0}} + \frac{Y - Y_{0}}{2\sqrt{Y_{0}}} \right) \left( \sqrt{Z_{0}} + \frac{Z - Z_{0}}{2\sqrt{Z_{0}}} \right)} = \frac{\frac{X}{\sqrt{Y_{0}} \sqrt{Z_{0}}}}{\left( 1 + \frac{Y - Y_{0}}{2Y_{0}} \right) \left( 1 + \frac{Z - Z_{0}}{2Z_{0}} \right)} .
\eea
Solving the product in the denominator on the right hand side last term and keeping only the terms up to the first order in the increment, we have
\bea
Corr_{R}(S^{1},S^{2}) \approx  \frac{\frac{X}{\sqrt{Y_{0}} \sqrt{Z_{0}}}}{1 + \left( \frac{Y - Y_{0}}{2Y_{0}} + \frac{Z - Z_{0}}{2Z_{0}} \right)} \approx \frac{X}{\sqrt{Y_{0}} \sqrt{Z_{0}}} \left[ 1 - \left( \frac{Y - Y_{0}}{2Y_{0}} + \frac{Z - Z_{0}}{2Z_{0}} \right) \right] .
\eea
In what follows, we will consider only the zero order of approximation, we discuss the first correction on appendix. Here, we are going to approximate the realized correlation as
\bea
Corr_{R}(S^{1},S^{2}) \approx  \frac{X}{\sqrt{Y_{0}} \sqrt{Z_{0}}}
\eea
Substituting $X$, $Y$ and $Z$ we obtain
\bea
\label{Eq:CorrApprox}
\begin{aligned}
Corr_{R}(S^{1},S^{2}) \approx \frac{1}{\sqrt{\Exp \{  \sigma^{(1)^{2}}_{R}(x) \}} \sqrt{\Exp \{  \sigma^{(2)^{2}}_{^{2}R}(x) \}}} \frac{1}{T}\int_{0}^{T} \rho_{t} \sigma^{(1)}(x_{t},\gamma(t)) \sigma^{(2)}(x_{t},\gamma(t)) dt 
\end{aligned}
\eea
where (cf. teorem \ref{teo:VarSwap}), we have
\bea
\Exp \left\{  \sigma^{(i)^{2}}_{R}(x) \right\} = \Exp \left\{  \frac{1}{T}\int_{0}^{T} (\sigma^{(i)}(x_{t},\gamma(t)))^{2} dt \right\} = \frac{1}{T}\int_{0}^{T} e^{tQ} (\sigma^{(i)}(x,0))^{2} dt ,
\eea
for $i = 1, 2$. In order to price a correlation swap we have to be able to evaluate the expectation of both side of equation (\ref{Eq:CorrApprox}), the expectation of the right hand side becomes
\bea
\begin{aligned}
\Exp \left\{ \sigma^{(1)}(x_{t},\gamma(t)) \sigma^{(2)}(x_{t},\gamma(t)) \right\} = e^{tQ} \sigma^{(1)}(x,0) \sigma^{(2)}(x,0) .
\end{aligned}
\eea
We can summarize the previous result in the following statement.
\begin{teo}
The value of a correlation swap for semi-Markov stochastic volatility is
\bea
P_{corr}(x) \approx e^{-rT} \left\{ \frac{ \int_{0}^{T} \rho_{t} e^{tQ} [\sigma^{(1)}(x,0) \sigma^{(2)}(x,0)] dt}{\sqrt{\int_{0}^{T} e^{tQ} (\sigma^{(1)}(x,0))^{2} dt} \sqrt{\int_{0}^{T} e^{tQ} (\sigma^{(2)}(x,0))^{2} dt}} - K_{corr} \right\} ,
\eea
where $Q$ is the generator of $(x_{t},\gamma(t))_{t}$, that is
\bea
Q f(x,t) = \frac{d f}{dt}(x,t) + \frac{g_{x}(t)}{\overline{G}_{x}(t)} \int_{X}P(x,dy) [f(y,0) - f(x,t)] .
\eea
\end{teo}

\section{Numerical Evaluation of Covariance and Correlation Swaps with Semi-Markov Stochastic Volatility}

In order to obtain a more handy expression for the price of covariance and correlation swaps to use in the application, we will introduce here an approximation for the family of operator $(e^{tQ})_{t}$. Following the approach used for the variance and volatility case, we are going to approximate the operators at the first order in Q as
\bea
e^{tQ} f(\cdot) \approx (I + tQ)f(\cdot).
\eea
Using this approximation the covariance swap price becomes
\bea
\begin{aligned}
P_{cov}(x) \approx & e^{-rT} \left\{ \frac{1}{T} \int_{0}^{T} \rho_{t} (I + tQ) [\sigma^{(1)}(x,0) \sigma^{(2)}(x,0)] dt - K_{cov} \right\} \\
= & e^{-rT} \left\{ \sigma^{(1)}(x,0) \sigma^{(2)}(x,0) \int_{0}^{T} \rho_{t} dt + Q [\sigma^{(1)}(x,0) \sigma^{(2)}(x,0)]  \int_{0}^{T} t \rho_{t} dt - K_{cov} \right\} .
\end{aligned}
\eea
The same approximation allow us to express the correlation swap price as
 \bea
\begin{aligned}
P_{corr}(x) \approx & e^{-rT} \left\{ \frac{ \int_{0}^{T} \rho_{t} (I + tQ) [\sigma^{(1)}(x,0) \sigma^{(2)}(x,0)] dt}{\sqrt{\int_{0}^{T} (I + tQ) (\sigma^{(1)}(x,0))^{2} dt} \sqrt{\int_{0}^{T} (I + tQ) (\sigma^{(2)}(x,0))^{2} dt}} - K_{corr} \right\} \\
= & e^{-rT} \left\{ \frac{ \sigma^{(1)}(x,0) \sigma^{(2)}(x,0)  \int_{0}^{T} \rho_{t} dt + Q [ \sigma^{(1)}(x,0) \sigma^{(2)}(x,0)] \int_{0}^{T} t \rho_{t}  dt}{\sqrt{(\sigma^{(1)}(x,0))^{2} + \frac{T}{2} Q(\sigma^{(1)}(x,0))^{2}} \sqrt{(\sigma^{(2)}(x,0))^{2} + \frac{T}{2} Q(\sigma^{(2)}(x,0))^{2}}} - K_{corr} \right\}.
\end{aligned}
\eea
 
\clearpage

{\Huge {\bf Appendix}}

\vspace{20pt}

\appendix

\section{Correlation Swaps: First Order Correction}

We would like to obtain an approximation for the realized correlation between two risky assets
bea
\bea
Corr_{R}(S^{1},S^{2}) = \frac{Cov_{R}(S^{1},S^{2})}{\sqrt{\sigma^{(1)^{2}}_{R}(x)}\sqrt{\sigma^{(2)^{2}}_{R} (x)}} .
\eea
In section \ref{Sec:CorrSimple} we have already obtained the following approximated expression
\bea
\label{Eq:App.1}
Corr_{R}(S^{1},S^{2}) \approx  \frac{\frac{X}{\sqrt{Y_{0}} \sqrt{Z_{0}}}}{1 + \left( \frac{Y - Y_{0}}{2Y_{0}} + \frac{Z - Z_{0}}{2Z_{0}} \right)} \approx \frac{X}{\sqrt{Y_{0}} \sqrt{Z_{0}}} \left[ 1 - \left( \frac{Y - Y_{0}}{2Y_{0}} + \frac{Z - Z_{0}}{2Z_{0}} \right) \right] .
\eea
where
\bea
\begin{aligned}
X = & Cov_{R}(S^{1},S^{2}) \\
Y = & \sigma^{(1)^{2}}_{R}(x) \\
Z = & \sigma^{(2)^{2}}_{R}(x) ,
\end{aligned}
\eea
and with the pedix $0$ we have denoted the expected values. We have already evaluated the expectation of the zero order approximation, now we would like to evaluate the first order. \\
Substituting $X$, $Y$ and $Z$ in equation (\ref{Eq:App.1}) we obtain
\bea
\label{Eq:App.2}
\begin{aligned}
Corr_{R}(S^{1},S^{2}) \approx & \frac{1}{\sqrt{\Exp \{  \sigma^{(1)^{2}}_{R}(x) \}} \sqrt{\Exp \{  \sigma^{(2)^{2}}_{^{2}R}(x) \}}} \frac{1}{T}\int_{0}^{T} \rho_{t} \sigma^{(1)}(x_{t},\gamma(t)) \sigma^{(2)}(x_{t},\gamma(t)) dt \\
- & \frac{1}{2 T^{2} (\Exp \{  \sigma^{(1)^{2}}_{R}(x) \} )^{3/2} (\Exp \{  \sigma^{(2)^{2}}_{R}(x) \})^{3/2}} \int_{0}^{T} \rho_{t} \sigma^{(1)}(x_{t},\gamma(t)) \sigma^{(2)}(x_{t},\gamma(t)) dt \\
\times & \left\{ \Exp \{  \sigma^{(2)^{2}}_{R}(x) \} \int_{0}^{T} [ (\sigma^{(1)}(x_{s},\gamma(s)))^{2} - \Exp \{ (\sigma^{(1)}(x_{s},\gamma(s)))^{2} \} ] ds \right. \\
+ & \left. \Exp \{  \sigma^{(1)^{2}}_{R}(x) \} \int_{0}^{T} [ (\sigma^{(2)}(x_{u},\gamma(u)))^{2} - \Exp \{ (\sigma^{(2)}(x_{u},\gamma(u)))^{2} \} ] du \right\} ,
\end{aligned}
\eea
where
\bea
\Exp \left\{  \sigma^{2}_{(i)R}(x) \right\} = \Exp \left\{  \frac{1}{T}\int_{0}^{T} (\sigma^{(i)}(x_{t},\gamma(t)))^{2} dt \right\} = \frac{1}{T}\int_{0}^{T} e^{tQ} (\sigma^{(i)}(x,0))^{2} dt ,
\eea
for $i = 1, 2$. We have to evaluate the expectation of the right hand side of equation (\ref{Eq:App.2}). We already calculated the expectation of the first term, which is the zero order approximation for the realized correlation.  Then we will focus now on the other terms. First of all, let rewrite them as follows
\bea
\begin{aligned}
\int_{0}^{T} \int_{0}^{T} & \rho_{t} \sigma^{(1)}(x_{t},\gamma(t)) \sigma^{(2)}(x_{t},\gamma(t))\left( \Exp \{  \sigma^{(2)^{2}}_{R}(x) \} [ (\sigma^{(1)}(x_{s},\gamma(s)))^{2} - \Exp \{ (\sigma^{(1)}(x_{s},\gamma(s)))^{2} \} ] \right. \\
+ & \left. \Exp \{  \sigma^{(1)^{2}}_{R}(x) \} [(\sigma^{(2)}(x_{s},\gamma(s)))^{2} - \Exp \{ (\sigma^{(2)}(x_{s},\gamma(s)))^{2} \} ] \right) ds dt ,
\end{aligned}
\eea
we have four different contributions in the integrals, the expectation of the terms
\bea
\int_{0}^{T} \int_{0}^{T} \rho_{t} \sigma^{(1)}(x_{t},\gamma(t)) \sigma^{(2)}(x_{t},\gamma(t)) \Exp \{ \sigma^{(i)^{2}}(x_{s},\gamma(s)) \} \Exp \{ \sigma^{(-i)^{2}}(x_{s},\gamma(s)) \} dsdt
\eea
for $i=1,2$, can be evaluate using theorem  \ref{teo:CovSwap}. Then, in order to evaluate the expectation of the approximated realized correlation, it only remains to calculate
\bea
\Exp \left\{ \int_{0}^{T} \int_{0}^{T} \rho_{t} \sigma^{(1)}(x_{t},\gamma(t)) \sigma^{(2)}(x_{t},\gamma(t)) \sigma^{(i)^{2}}(x_{s},\gamma(s)) dsdt \right\} \qquad \qquad i=1,2
\eea
To this end, let's first divide the range of integration in two intervals as follows
\bea
\begin{aligned}
& \Exp \left\{ \int_{0}^{T} \int_{0}^{t} \rho_{t} \sigma^{(1)}(x_{t},\gamma(t)) \sigma^{(2)}(x_{t},\gamma(t)) \sigma^{(i)^{2}}(x_{s},\gamma(s)) dsdt \right. \\
+ & \left. \int_{0}^{T} \int_{t}^{T} \rho_{t} \sigma^{(1)}(x_{t},\gamma(t)) \sigma^{(2)}(x_{t},\gamma(t)) \sigma^{(i)^{2}}(x_{s},\gamma(s)) dsdt \right\}
\end{aligned}
\eea
for $i = 1,2$. We notice that the first integral set is such that $t>s$ while the second has $t<s$. We can now use the property of conditional expectation to obtain
\bea
\label{Eq:App.3}
\begin{aligned}
& \Exp \left\{ \int_{0}^{T} \int_{0}^{t} \rho_{t} \Exp \{ \sigma^{(1)}(x_{t},\gamma(t)) \sigma^{(2)}(x_{t},\gamma(t)) | \Alg_{s} \} \sigma^{(i)^{2}}(x_{s},\gamma(s)) dsdt \right. \\
+ & \left. \int_{0}^{T} \int_{t}^{T} \rho_{t} \sigma^{(1)}(x_{t},\gamma(t)) \sigma^{(2)}(x_{t},\gamma(t)) \Exp \{ \sigma^{(i)^{2}}(x_{s},\gamma(s)) | \Alg_{t} \} dsdt \right\} .
\end{aligned}
\eea
We notice that $(x_{t},\gamma(t))_{t}$ is a Markov process then using the Markov property, we can express the conditional expectations as
\bea
\Exp \{ \sigma^{(1)}(x_{t},\gamma(t)) \sigma^{(2)}(x_{t},\gamma(t)) | \Alg_{s} \} = e^{(t-s)Q}\sigma^{(1)}(x_{s},\gamma(s)) \sigma^{(2)}(x_{s},\gamma(s)) =: h(x_{s},\gamma(s)) \nonumber
\eea
for $t>s$, and
\bea
 \Exp \{ \sigma^{(i)^{2}}(x_{s},\gamma(s)) | \Alg_{t} \} = e^{(s-t)Q} \sigma^{(i)^{2}}(x_{t},\gamma(t)) =: g^{(i)}(x_{t},\gamma(t)) \nonumber
 \eea
 for $s>t$. Therefore, the first term of eq. (\ref{Eq:App.3}) can be expressed as
  \bea
 \Exp \left\{ \int_{0}^{T} \int_{0}^{t} \rho_{t} h(x_{s},\gamma(s)) \sigma^{(i)^{2}}(x_{s},\gamma(s)) dsdt \right\} = \int_{0}^{T} \int_{0}^{t} \rho_{t} e^{sQ} [ h(x,0) \sigma^{(i)^{2}}(x,0)] dsdt , 
\eea 
while the second as
\bea
\begin{aligned}
& \Exp \left\{ \int_{0}^{T} \int_{t}^{T} \rho_{t} \sigma^{(1)}(x_{t},\gamma(t)) \sigma^{(2)}(x_{t},\gamma(t)) g^{(i)}(x_{t},\gamma(t)) dsdt \right\} \\
 = & \int_{0}^{T} \int_{t}^{T} \rho_{t} e^{tQ} [ \sigma^{(1)}(x,0) \sigma^{(2)}(x,0) g^{(i)}(x,0) ] dsdt .
 \end{aligned}
\eea
Now, we can evaluate the functions $h$ and $g$ at x obtaining 
\bea
h(x,0) = e^{(t-s)Q} [ \sigma^{(1)}(x,0) \sigma^{(2)}(x,0) ]
\eea
and
\bea
g^{(i)}(x,0) = e^{(s-t)Q} [ \sigma^{(i)^{2}}(x,0) ] .
\eea
We can summarize the previous result in the following statement which gives the correlation swap price up to the first order of approximation.
\begin{teo}
The value of the correlation swap for a semi-Markov volatility is
\bea
P_{corr}(x) = e^{-rT} \left( \Exp \{ Corr_{R}(S^{1},S^{2}) \} - K_{corr} \right)
\eea
where the realized correlation can be approximated by
\bea
\begin{aligned}
& \hspace{-20pt}\Exp \{ Corr_{R}(S^{1},S^{2}) \} \approx  \frac{ 2 \int_{0}^{T} \rho_{t} e^{tQ} \sigma^{(1)}(x,0) \sigma^{(2)}(x,0) dt}{\sqrt{\int_{0}^{T} e^{tQ} (\sigma^{(1)}(x,0))^{2} dt} \sqrt{\int_{0}^{T} e^{tQ} (\sigma^{(2)}(x,0))^{2} dt}} \\
& \hspace{-20pt} - \frac{\int_{0}^{T}\rho_{t} ( \int_{0}^{t} e^{sQ} \{ e^{tQ} [ \sigma^{(1)}(x,0) \sigma^{(2)}(x,0) ]   \sigma^{(1)^{2}}(x,0) \} ds + \int_{t}^{T} e^{tQ} \{ \sigma^{(1)}(x,0) \sigma^{(2)}(x,0) e^{uQ} [ \sigma^{(1)^{2}}(x,0) ] \} du)dt }{2 \left( \int_{0}^{T} e^{tQ} (\sigma^{(1)}(x,0))^{2} dt \right)^{3/2} \left( \int_{0}^{T} e^{tQ} (\sigma^{(2)}(x,0))^{2} dt \right)^{1/2}} \\
& \hspace{-20pt} - \frac{\int_{0}^{T}\rho_{t} ( \int_{0}^{t} e^{sQ} \{ e^{tQ} [ \sigma^{(1)}(x,0) \sigma^{(2)}(x,0) ]   \sigma^{(2)^{2}}(x,0) \} ds 
+ \int_{t}^{T} e^{tQ} \{ \sigma^{(1)}(x,0) \sigma^{(2)}(x,0) e^{uQ} [ \sigma^{(2)^{2}}(x,0) ] \} du ) dt }{2 \left( \int_{0}^{T} e^{tQ} (\sigma^{(1)}(x,0))^{2} dt \right)^{1/2} \left( \int_{0}^{T} e^{tQ} (\sigma^{(2)}(x,0))^{2} dt \right)^{3/2}} ,
\end{aligned}
\eea
here Q is the generator of the Markov process $(x_{t},\gamma(t))_{t}$ given by
\bea
Q f(x,t) = \frac{d f}{dt}(x,t) + \frac{g_{x}(t)}{\overline{G}_{x}(t)} \int_{X}P(x,dy) [f(y,0) - f(x,t)] .
\eea
\end{teo}

\clearpage


\end{document}